\documentclass[11pt]{article}
\usepackage{amsthm} 
\usepackage{amsmath} 
\usepackage{tikz}

\newtheorem{theorem}{Theorem}
\newtheorem{lemma}{Lemma}
\newtheorem{proposition}{Proposition}

\newtheorem{corollary}{Corollary}

{
\theoremstyle{definition}
\newtheorem{definition}[theorem]{Definition}
}

\newcommand{\Qs}[1]{\ensuremath{Q^{*}(#1)}}
\newcommand{\Free}{\ensuremath{\mathrm{Free}}}

\tikzstyle{vertex}=[circle,draw, fill,minimum size=6pt,inner sep=0pt]
\tikzstyle{edge} = [draw,thick,-]
\tikzstyle{weight} = [fill=blue!3,font=\scriptsize]

\date{}
\title{Upper Domination: \\ towards a dichotomy through boundary properties\thanks{\tiny The results of this paper previously appeared as 
extended abstracts in proceedings of the 8th International Conference on Combinatorial Optimization and Applications,  COCOA 2014 \cite{monogenic}
and the 27th International Workshop on Combinatorial Algorithms, IWOCA 2016 \cite{iwoca}.}}

\author{Hassan AbouEisha\thanks{\tiny King Abdullah University of Science and Technology, Thuwal, Saudia Arabia. Email: hassan.aboueisha@kaust.edu.sa} \hspace{3mm}
Shahid Hussain\thanks{\tiny Habib University, Karachi, Pakistan. Email: shahid.hussain@sse.habib.edu.pk} \hspace{3mm}
Vadim Lozin\thanks{\tiny Mathematics Institute, University of Warwick, Coventry, CV4 7AL, UK. Email: V.Lozin@warwick.ac.uk 
} \\
J\'er\^ome Monnot\thanks{\tiny Universit\'e~Paris-Dauphine,~PSL~Research~University,~CNRS,~LAMSADE,~75016~Paris,~France.~Email:~jerome.monnot@dauphine.fr} \hspace{4mm}
Bernard Ries\thanks{\tiny University~of~Fribourg,~Department~of~Informatics,~Bd~de~P\'erolles~90,~1700~Fribourg,~Switzerland.~Email:~bernard.ries@unifr.ch} \hspace{4mm}
Viktor Zamaraev\thanks{\tiny Mathematics Institute, University of Warwick, Coventry, CV4 7AL, UK. Email: V.Zamaraev@warwick.ac.uk}
}

\begin{document}
\maketitle

\begin{abstract}
An upper dominating set in a graph is a minimal (with respect to set inclusion) dominating set of maximum cardinality. 
The problem of finding an upper dominating set is generally NP-hard. We study the complexity of this problem in 
classes of graphs defined by  finitely many forbidden induced subgraphs and conjecture that the problem admits 
a dichotomy in this family, i.e. it is either NP-hard or polynomial-time solvable for each class in the family.
A helpful tool to study the complexity of an algorithmic problem on finitely defined classes of graphs is the notion of 
boundary classes. However, none of such classes has been identified so far for the upper dominating set problem.  
In the present paper, we discover the first boundary class for this problem and prove the dichotomy for classes 
defined by a single forbidden induced subgraph. 
\end{abstract}

\section{Introduction}

In a graph $G=(V,E)$, a {\it dominating set} is a subset of vertices $D\subseteq V$ such that any vertex
outside of $D$ has a neighbour in $D$. A dominating set $D$ is {\it minimal} if no proper subset of $D$ 
is dominating. An {\it upper dominating set} is a minimal dominating set of maximum cardinality. 
The {\sc upper dominating set} problem (i.e. the problem of finding an upper dominating set in a graph) is generally NP-hard \cite{ChestonFHJ90}. 
Moreover, the problem is difficult from a parameterized perspective (it is W[2]-hard \cite{Jerome1}) and from an approximation point of view
(for any $\varepsilon> 0$, the problem is not $n^{1-\varepsilon}$-approximable, unless $P=NP$ \cite{Jerome2}). 
On the other hand, in some particular graph classes the problem can be solved in polynomial time, which is the case for 
bipartite graphs \cite{CockayneFPT81}, chordal graphs \cite{JacobsonP90}, generalized series-parallel graphs \cite{HareHLPPW87}, graphs of bounded clique-width \cite{Courcelle}, etc.
We contribute to this topic in several ways.

First, we prove two new NP-hardness results: for complements of bipartite graphs and for planar graphs of vertex degree at most 6 and girth at least 6.
This leads to a complete dichotomy for this problem in the family of minor-closed graph classes. 
Indeed, if a minor-closed class $X$ contains all planar graphs, then the problem is NP-hard in $X$.
Otherwise, graphs in $X$ have bounded tree-with \cite{excluding} (and hence bounded clique-width), in which case the problem can be solved in polynomial time.
Whether this dichotomy can be extended to the family of all hereditary classes is a challenging open question. 
We conjecture that the classes defined by finitely many forbidden induced subgraphs admit such a dichotomy and prove 
several results towards this goal. For this, we employ the notion of boundary classes that has been recently introduced 
to study algorithmic graph problems. The importance of this notion is due to the fact that an algorithmic problem $\Pi$ 
is NP-hard in a class $X$ defined by finitely many forbidden induced subgraphs if and only if $X$ contains a boundary class for $\Pi$.
Unfortunately, no boundary class is known for the {\sc upper dominating set} problem. In the present paper, we unveil this uncertainty 
by discovering the first boundary class for the problem. We also develop a polynomial-time algorithm for {\sc upper domination}
in the class of $2K_2$-free graphs. Combining various results of the paper, we prove that the dichotomy holds for classes 
defined by a single forbidden induced subgraph.

The organization of the paper is as follows. In Section~\ref{sec:pre}, we introduce basic definitions, including the notion of a boundary class, 
and prove some preliminary results. In Section~\ref{sec:NP}, we prove two NP-hardness results. 
Section~\ref{sec:boundary} is devoted to the first boundary class for the problem.  In Section~\ref{sec:poly},
we establish the dichotomy for classes defined by a single forbidden induced subgraph.   Finally, Section~\ref{sec:con} concludes the paper with a number of open problems. 
  
\section{Preliminaries}
\label{sec:pre}

We denote by $\cal G$ the set of all simple graphs, i.e. undirected graphs without loops and multiple edges.
The {\it girth} of a graph $G\in {\cal G}$ is the length of a shortest cycle in $G$.
As usual, $K_n$, $P_n$ and $C_n$ stand for the complete graph, the chordless path and the chordless cycle with $n$ vertices, respectively.
Also, $\overline{G}$ denotes the complement of $G$, and $2K_2$ is the disjoint union of two copies of $K_2$.
A {\it star} is a connected graph in which all edges are incident to the same vertex, called the \textit{center} of the star. 

Let $G=(V,E)$ be a graph with vertex set $V$ and edge set $E$, and let $u$ and $v$ be two vertices of $G$. 
If  $u$ is adjacent to  $v$, we write $uv\in E$ and say that $u$ and $v$ are \textit{neighbours}.
The \textit{neighbourhood} of a vertex $v\in V$ is the set of its neighbours; it is denoted by $N(v)$.
The \textit{degree} of $v$ is the size of its neighbourhood. 
If the degree of each vertex of $G$ equals 3, then $G$ is called {\it cubic}.

A subgraph of $G$ is {\it spanning} if it contains all vertices of $G$, and it is 
{\it induced} if two vertices of the subgraph are adjacent if and only if they are adjacent in $G$.
If a graph $H$ is isomorphic to an induced subgraph of a graph $G$, we say that~$G$ contains~$H$. 
Otherwise we say that $G$ is $H$-free. Given a set of graphs $M$, we denote by $\Free(M)$ the set of 
all graphs containing no induced subgraphs from $M$.

A class of graphs (or graph property) is a set of graphs closed under isomorphism.
A class is {\it hereditary} if it is closed under taking induced subgraphs. It is well-known 
(and not difficult to see) that a class $X$ is hereditary if and only if $X=\Free(M)$ for some set $M$.
If $M$ is a finite set, we say that $X$ is {\it finitely defined}, and if $M$ consists of a single graph, 
then $X$ is {\it monogenic}. 

A class of graphs is {\it monotone} if it is closed under taking subgraphs (not necessarily induced).
Clearly, every monotone class is hereditary.

In a graph, a {\it clique} is a subset of pairwise adjacent vertices, and 
an {\it independent set} is a subset of vertices no two of which are adjacent.
A graph is {\it bipartite} if its vertices can be partitioned into two independent sets.
It is well-known that a graph is bipartite if and only if it is free of odd cycles, 
i.e. if and only if it belongs to $\Free(C_3,C_5,C_7,\ldots)$.  
We say that a graph $G$ is {\it co-bipartite} if $\overline{G}$ is bipartite. 
Clearly, a graph is co-bipartite if and only if it belongs to $\Free(\overline{C}_3,\overline{C}_5,\overline{C}_7,\ldots)$.

We say that an independent set $I$ is {\it maximal} if no other independent set properly 
contains $I$. The following simple lemma connects the notion of a maximal independent set 
and that of a minimal dominating set. 

\begin{lemma}\label{lem:mm}
Every maximal independent set is a minimal dominating set.
\end{lemma}

\begin{proof}
Let $G=(V,E)$ be a graph and let $I$ be a maximal independent set in~$G$. 
Then every vertex $u\not\in I$ has a neighbour in $I$ (else $I$ is not maximal) and hence~$I$ is dominating. 

The removal of any vertex $u\in I$ from $I$ leaves $u$ undominated.
Therefore, $I$ is a minimal dominating set.  
\end{proof}

\begin{definition}
Given a dominating set $D$ and a vertex $x\in D$, we say that a vertex $y\not\in D$ is a {\it private}
neighbour of $x$ if $x$ is the only neighbour of $y$ in $D$.
\end{definition}


\begin{lemma}\label{lem:1}
Let $D$ be a minimal dominating set in a graph $G$. If a vertex $x\in D$ has a neighbour in $D$, then 
it also has a private neighbour outside of $D$.
\end{lemma}

\begin{proof}
If a vertex $x\in D$ is adjacent to a vertex in $D$ and has no private neighbour outside of $D$,
then $D$ is not minimal, because the set $D-\{x\}$ is also dominating. 
\end{proof}

\begin{lemma}\label{lem:2}
Let $G$ be a connected graph and $D$ a minimal dominating set in $G$. If there are vertices in $D$ 
that have no private neighbour outside of $D$, then $D$ can be transformed in polynomial time
into a minimal dominating set $D'$ with $|D'|\le |D|$ in which every vertex has a private neighbour outside of $D'$.
\end{lemma}

\begin{proof}
Assume $D$ contains a vertex $x$ which has no private neighbours outside of $D$. Then $x$ is isolated in $D$
(i.e. it has no neighbours in $D$) by Lemma~\ref{lem:1}. On the other hand, since $G$ is connected, $x$ must have a neighbour $y$ outside of $D$.
As $y$ is not a private neighbour of $x$, it is adjacent to a vertex $z$ in $D$.
Consider now the set $D_0=(D-\{x\})\cup\{y\}$. Clearly, it is a dominating set.
If it is a minimal dominating set in which every vertex has a private neighbour outside of the set, then we are done.
Otherwise, it is either not minimal, in which case we can reduce its size by deleting some vertices,
or it has strictly fewer isolated vertices than $D$. Therefore, by iterating the procedure, in at most $|V(G)|$ 
steps we can transform  $D$ into a minimal dominating set $D'$ with $|D'|\le |D|$ in which every vertex has a private neighbour outside of the set. 
\end{proof}

\subsection{Boundary classes of graphs}
\label{sec:b-classes}

The notion of boundary classes of graphs was introduced in \cite{Ale03} to study the {\sc maximum independent set}
problem in hereditary classes. Later this notion was applied to some other  
problems of both algorithmic \cite{AKL04,ABKL07,KLMT11,LP13} and combinatorial \cite{KLR13,Loz08,LZ15} nature.
Assuming $P\ne NP$, the notion of boundary classes can be defined, with respect to algorithmic graph problems, as follows. 

\medskip
Let $\Pi$ be an algorithmic graph problem, which is generally NP-hard. We will say that a hereditary class $X$ of graphs
is $\Pi$-{\it tough} if the problem is NP-hard for graphs in $X$ and $\Pi$-{\it easy}, otherwise. 
We define the notion of a boundary class for $\Pi$ in two steps. First, let us define the notion of a limit class.

\begin{definition}
A hereditary class $X$ is a {\it limit class} for $\Pi$ if $X$ is the intersection of a sequence $X_1\supseteq X_2\supseteq X_3\supseteq \ldots$
of $\Pi$-tough classes, in which case we also say that the sequence {\it converges} to $X$.
\end{definition} 

\begin{itemize}
\item[]
{\it Example}. To illustrate the notion of a limit class, let us quote a result from 
\cite{Mur} stating that the {\sc maximum independent set} problem is NP-hard for graphs with large girth, i.e. for $(C_3,C_4,\ldots,C_k)$-free graphs 
for each fixed value of $k$. With $k$ tending to infinity, this sequence converges to the class of graphs without cycles, i.e. to forests.
Therefore, the class of forests is a limit class for the {\sc maximum independent set} problem.  However, this is not a minimal limit class for 
the problem, which can be explained as follows. 

The proof of the NP-hardness of the problem for graphs with large girth is based on a simple fact 
that a double subdivision of an edge in a graph $G$ increases the size of a maximum independent set in $G$ by exactly 1. This operation applied 
sufficiently many (but still polynomially many)  times allows to destroy all small cycles in $G$, i.e. reduces the problem from an arbitrary graph $G$
to a graph $G'$ of girth at least $k$. Obviously, if $G$ is a graph of vertex degree at most 3, then so is $G'$, and since the problem is NP-hard
for graphs of degree at most 3, we conclude that it is also NP-hard for for $(C_3,C_4,\ldots,C_k)$-free graphs of degree at most 3. This shows 
that the class of forests of vertex degree at most 3 is a limit class for the the {\sc maximum independent set} problem. However, it is still 
not a minimal limit class, because by the same operation (double subdivisions of edges) one can destroy small induced copies of the graph $H_n$
shown on the left of Figure~\ref{fig:SH}. 
Therefore, the {\sc maximum independent set} problem is NP-hard in the following class for each fixed value of $k$: 

\medskip
\begin{itemize}
\item[$Z_k$] is the class of  $(C_3,\ldots,C_k,H_1,\ldots,H_k)$-free graphs of degree at most 3.
\end{itemize}
\medskip

It is not difficult to see that the sequence $Z_3\supset Z_4\supset\ldots$ converges to the class of forests every connected component of which 
has the form $S_{i,j,\ell}$ represented on the right of Figure~\ref{fig:SH}, also known as {\it tripods}. Throughout the paper we denote this class by $\cal S$, i.e.

\medskip
\begin{itemize}
\item[$\cal S$] is the intersection of the sequence $Z_3\supset Z_4\supset\ldots$.
\end{itemize}
\medskip

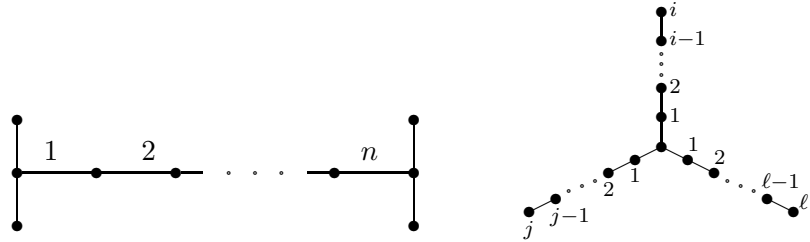
\begin{figure}[ht]
\begin{center}
\begin{picture}(170,70)
\put(-10,25){\circle*{4}} 
\put(20,25){\circle*{4}}
\put(50,25){\circle*{4}} 
\put(70,25){\circle{1}}
\put(80,25){\circle{1}} 
\put(90,25){\circle{1}}
\put(110,25){\circle*{4}} 
\put(140,25){\circle*{4}}
\put(-10,45){\circle*{4}} 
\put(-10,5){\circle*{4}}
\put(140,45){\circle*{4}} 
\put(140,5){\circle*{4}}
\put(-8,25){\line(1,0){26}} 
\put(22,25){\line(1,0){26}}
\put(52,25){\line(1,0){8}} 
\put(100,25){\line(1,0){8}}
\put(112,25){\line(1,0){26}} 
\put(-10,27){\line(0,1){16}}
\put(-10,23){\line(0,-1){16}} 
\put(140,27){\line(0,1){16}}
\put(140,23){\line(0,-1){16}} 
\put(0,30){1} \put(37,30){2}
\put(120,30){$n$}
\end{picture}
\begin{picture}(120,80)
\put(60,35){\circle*{4}} \put(60,46){\circle*{4}}
\put(60,57){\circle*{4}} \put(60,75){\circle*{4}}
\put(60,86){\circle*{4}} \put(60,62){\circle{1}}
\put(60,66){\circle{1}} \put(60,70){\circle{1}}
\put(60,35){\line(0,1){11}} \put(60,46){\line(0,1){11}}
\put(60,75){\line(0,1){11}} \put(50,30){\circle*{4}}
\put(40,25){\circle*{4}} \put(20,15){\circle*{4}}
\put(10,10){\circle*{4}} \put(35,22){\circle{1}}
\put(30,20){\circle{1}} \put(25,18){\circle{1}}
\put(60,35){\line(-2,-1){10}} \put(50,30){\line(-2,-1){10}}
\put(20,15){\line(-2,-1){10}} \put(70,30){\circle*{4}}
\put(80,25){\circle*{4}} \put(100,15){\circle*{4}}
\put(110,10){\circle*{4}} \put(85,22){\circle{1}}
\put(90,20){\circle{1}} \put(95,18){\circle{1}}
\put(60,35){\line(2,-1){10}} \put(70,30){\line(2,-1){10}}
\put(100,15){\line(2,-1){10}}

\put(63,46){$_1$} \put(63,57){$_2$} \put(63,75){$_{i-1}$}
\put(63,86){$_i$} \put(48,23){$_1$} \put(38,18){$_2$}
\put(18,8){$_{j-1}$} \put(8,3){$_j$} \put(70,35){$_1$}
\put(80,30){$_2$} \put(98,21){$_{\ell-1}$} \put(112,13){$_\ell$}
\end{picture}
\end{center}
\caption{Graphs $H_n$ (left) and $S_{i,j,\ell}$ (right).} 
\label{fig:SH}
\end{figure}

The above discussion shows that $\cal S$ is a limit class for the {\sc maximum independent set} problem. Moreover, in \cite{Ale03}
it was proved that $S$ is a {\it minimal} limit class for this problem. 
\end{itemize}

\begin{definition}
A minimal (with respect to set inclusion) limit class for a problem $\Pi$ is called a {\it boundary class} for $\Pi$.
\end{definition} 

The importance of the notion of boundary classes for NP-hard algorithmic graph problems is due to the following theorem proved originally  
for the {\sc maximum independent set} problem in \cite{Ale03} (can also be found in \cite{AKL04} in a more general context). 

\begin{theorem}\label{thm:boundary}
If $\mathrm{P}\neq \mathrm{NP}$, then a finitely defined class $X$ is $\Pi$-tough
if and only if $X$ contains a boundary class for $\Pi$. 
\end{theorem}

In Section~\ref{sec:boundary}, we identify the first boundary class for the {\sc upper dominating set} problem. To this end, we need a number of auxiliary results. 
The first of them is the following lemma dealing with limit classes, which was derived in \cite{Ale03,AKL04} as a step towards the proof of Theorem~\ref{thm:boundary}.   

\begin{lemma}\label{lem:limit}
If $X$ is a finitely defined class containing a limit class for an NP-hard problem $\Pi$, then $X$ is $\Pi$-tough.  
\end{lemma}

The next two results were proved in \cite{Korobitsyn} and \cite{AKL04}, respectively.

\begin{lemma}\label{lem:Zk}
The {\sc minimum dominating set} problem is NP-hard in the class $Z_k$ for each fixed value of $k$.
\end{lemma}

\begin{theorem}\label{thm:domS}
The class $\cal S$ is a boundary class for the {\sc minimum dominating set} problem.
\end{theorem}

\section{NP-hardness results}
\label{sec:NP}

In this section, we prove two NP-hardness results about the {\sc upper dominating set} problem in restricted graph classes.  

\subsection{Planar graphs of degree at most $6$ and girth at least $6$}

\begin{theorem}\label{thm:1}
The {\sc upper dominating set} problem restricted to the class of planar graphs with maximum vertex degree $6$ and girth at least $6$ is NP-hard.
\end{theorem}

\begin{proof}
We use a reduction from the {\sc maximum independent set} problem 
({\sc IS} for short) in planar cubic graphs, where {\sc IS} is NP-hard \cite{GareyJS76}. 
The input of the decision version of {\sc IS} consists of a simple graph $G=(V,E)$ and an integer $k$
and asks to decide if $G$ contains an independent set of size at least $k$. 

Let $G=(V,E)$ and an integer $k$ be an instance of {\sc IS}, where $G$ is a planar cubic graph. 
We denote the number of vertices and edges of $G$ by $n$ and $m$, respectively. 
We build an instance $G'=(V',E')$ of the {\sc upper dominating set} problem by replacing each edge $e=uv\in E$
with two induced paths $u-v_e-u_e-v$ and $u-v'_e-u'_e-v$, as shown in Figure~\ref{fig:1}. 

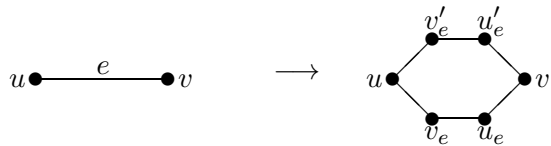
\begin{figure}[ht]
\begin{center}
\begin{picture}(250,60)


\put(25,30){\circle*{5}}
\put(75,30){\circle*{5}}

\put(160,30){\circle*{5}}
\put(210,30){\circle*{5}}
\put(175,15){\circle*{5}}
\put(195,15){\circle*{5}}
\put(175,45){\circle*{5}}
\put(195,45){\circle*{5}}

\put(25,30){\line(1,0){50}}
\put(160,30){\line(1,1){15}}
\put(160,30){\line(1,-1){15}}
\put(175,15){\line(1,0){20}}
\put(175,45){\line(1,0){20}}
\put(195,45){\line(1,-1){15}}
\put(195,15){\line(1,1){15}}

\put(15,27){$u$}
\put(79,27){$v$}
\put(48,32){$e$}

\put(150,27){$u$}
\put(214,27){$v$}
\put(172,7){$v_e$}
\put(192,7){$u_e$}

\put(172,49){$v'_e$}
\put(192,49){$u'_e$}

\put(115,30){$\longrightarrow$}
\end{picture}
\end{center}
\caption{Replacement of an edge by two paths}
\label{fig:1}
\end{figure}

Clearly, $G'$  can be constructed in time polynomial in $n$. 
Moreover, it is not difficult to see that $G'$ is a planar graph with maximum vertex degree 6 and girth at least $6$.

\medskip
We claim that $G$ contains an independent set of size at least $k$ if and only if~$G'$ contains a minimal dominating set of size at least $k+2m$.

\medskip
Suppose $G$ contains an independent set $S$ with $|S|\geq k$
and without loss of generality assume that $S$ is maximal with respect to set-inclusion 
(otherwise, we greedily add vertices to $S$ until it becomes a maximal independent
set). Now we consider a set $D\subset V'$ containing 
\begin{itemize}
\item all vertices of $S$,
\item vertices $v_e$ and $v'_e$ for each edge $e=uv\in E$ with $v\in S$,
\item exactly one vertex in $\{u_e,v_e\}$ (chosen arbitrarily) and exactly one vertex in $\{u'_e,v'_e\}$ (chosen arbitrarily) 
for each edge $e=uv\in E$ with $u,v\not\in S$.  
\end{itemize}
It is not difficult to see that $D$ is a maximal independent, and hence, by Lemma~\ref{lem:1}, a minimal dominating, set in $G'$.
Moreover, $|D|=|S|+2m\geq k+2m$.

\medskip
To prove the inverse implication, we first observe the following:
\begin{itemize}
\item {\it Every minimal dominating set in $G'$ contains either exactly two vertices or no vertex in the set $\{u_e,v_e,u'_e,v'_e\}$ for every
edge $e=uv\in E$}. Indeed, assume a minimal dominating set $D$ in $G'$ contains at least three vertices in $\{u_e,v_e,u'_e,v'_e\}$, 
say $u_e,v_e,u'_e$. But then $D$ is not minimal, since $u_e$ can be removed from the set. If $D$ contains one vertex in $\{u_e,v_e,u'_e,v'_e\}$, 
say $u_e$, then both $u$ and $v$ must belong to $D$ (otherwise it is not dominating), in which case it is not minimal ($u_e$ can be removed).
\item {\it If a minimal dominating set $D$ in $G'$ contains exactly two vertices in $\{u_e,v_e,u'_e,v'_e\}$, then} 
\begin{itemize}
\item {\it one of them belongs to  $\{u_e,v_e\}$ and the other to $\{u'_e,v'_e\}$}. Indeed, if both vertices belong to $\{u_e,v_e\}$, then 
both $u$ and $v$ must also belong to $D$ (to dominate $u'_e,v'_e$), in which case $D$ is not minimal ($u_e$ and $v_e$ can be removed). 
\item {\it at most one of $u$ and $v$ belongs to $D$}. Indeed, if both of them belong to $D$, then $D$ is not minimal dominating, 
because $u$ and $v$ dominate the set $\{u_e,v_e,u'_e,v'_e\}$  and any vertex of this set can be removed from $D$. 
\end{itemize}
\end{itemize}

Now let $D\subseteq V'$ be a minimal dominating set in $G'$ with $|D|\geq k+2m$. If $D$ contains exactly two vertices in the set 
$\{u_e,v_e,u'_e,v'_e\}$ for every edge $e=uv\in E$, then, according to the discussion above, the set $D\cap V$ is independent in $G$ 
and contains at least $k$ vertices, as required. 

Assume now that there are edges $e=uv\in E$ for which the set $\{u_e,v_e,u'_e,v'_e\}$ contains no vertex of $D$.
We call such edges {\it $D$-clean}. Obviously, both endpoints of a $D$-clean edge belong 
to $D$, since otherwise this set is not dominating. To prove the theorem in the situation when $D$-clean edges are present, 
we transform~$D$ into another minimal dominating set $D'$ with no $D'$-clean edges and with \hbox{$|D'|\ge |D|$}.
To this end, we do the following. For each vertex $u\in V$ incident to at least one $D$-clean edge, 
we first remove $u$ from $D$, and then for each $D$-clean edge \hbox{$e=uv\in E$} incident to $u$, we introduce vertices $v_e,v'_e$ to $D$. 
Under this transformation vertex~$v$ may become redundant (i.e. its removal may result in a dominating set), in which case
we remove it. It is not difficult to see that the set $D'$ obtained in this way is a minimal dominating set with no $D'$-clean edges and with $|D'|\ge |D|$.
Therefore, $D'\cap V$ is an independent set in $G$ of cardinality at least $k$. 
\end{proof}

\subsection{Complements of bipartite graphs}

To prove one more NP-hardness result for the {\sc upper dominating set} problem, let us introduce the following graph transformations. 
Given a graph $G=(V,E)$, we denote by 
\begin{itemize}
\item[$S(G)$] the incidence graph of $G$, i.e. the graph with vertex set $V\cup E$, where $V$ and $E$ are independent sets and a vertex $v\in V$  is adjacent to 
a vertex $e\in E$ in $S(G)$ if and only if $v$ is incident to $e$ in $G$. Alternatively, 
$S(G)$ is obtained from $G$ by subdividing each edge $e$ by a new vertex $v_e$. According to this interpretation,
we call $E$ the set of {\it new} vertices and $V$ the set of {\it old} vertices. Any graph of the form $S(G)$ for some $G$ will be called 
a {\it subdivision graph}. 
\item[$Q(G)$] the graph obtained from $S(G)$ by creating a clique on the set of old vertices and a clique on the set of new vertices.
We call any graph of the form $Q(G)$ for some $G$ a $Q$-{\it graph}.  
\end{itemize}

The importance of $Q$-graphs for the {\sc upper dominating set} problem is due to the following lemma, where we denote by $\Gamma(G)$
the size of an upper dominating set in $G$ and by $\gamma(G)$ the size of a dominating set of minimum cardinality in $G$.

\begin{lemma}\label{lem:Q}
Let $G$ be a graph with $n$ vertices such that $\Gamma(Q(G))\ge 3$. Then $\Gamma(Q(G))=n-\gamma(G)$.
\end{lemma}   

\begin{proof}
Let $D$ be a minimum dominating set in $G$, i.e. a dominating set of size $\gamma(G)$. Without loss of generality, 
we will assume that $D$ satisfies Lemma~\ref{lem:2}, i.e. every vertex of $D$ has a private neighbour outside of $D$. 
For every vertex $u$ outside of~$D$, consider exactly one edge, chosen arbitrarily, connecting $u$ to a vertex in $D$ and denote this edge by $e_u$.
We claim that the set \hbox{$D' = \{v_{e_u}\ :\  u \not\in D\}$} is a minimal dominating set in  $Q(G)$. 
By construction, $D'$ dominates \hbox{$E\cup (V - D)$} in $Q(G)$. To show that it also dominates $D$, assume by contradiction that 
a vertex $w\in D$ is not dominated by $D'$ in $Q(G)$. By Lemma~\ref{lem:2} we know that $w$ has a private neighbour $u$ outside of $D$.  
But then the edge $e=uw$ is the only edge connecting $u$ to a vertex in $D$. Therefore, $v_e$ necessarily belongs to $D'$ and
hence it dominates $w$, contradicting our assumption. In order to show that $D'$ is a minimal dominating set, we observe that 
if we remove from $D'$ a vertex~$v_{e_u}$ with $e_u=uw$, $u\not \in D$, $w\in D$, then $u$ becomes undominated in $Q(G)$. 
Finally, since $|D'| = n - |D|$, we conclude that  $\Gamma(Q(G))\ge n-|D|=n-\gamma(G)$.

\medskip
Conversely, let $D'$ be an upper dominating set in $Q(G)$, i.e. a minimal dominating set of size $\Gamma(Q(G))\ge 3$. 
Then $D'$ cannot intersect both $V$ and $E$, since otherwise it contains exactly one vertex in each of these sets (else it is not minimal, 
because each of these sets is a clique), in which case $|D'|=2$. 

Assume first that $D'\subseteq V$. Then $V-D'$ is an independent set in $G$. Indeed, if $G$ contains an edge $e$ connecting two vertices in
$V-D'$, then vertex $v_e$ is not dominated by $D'$ in $Q(G)$, a contradiction. Moreover, $V-D'$ is a maximal (with respect to set-inclusion) independent set 
in $G$, because $D'$ is a minimal dominating set in $Q(G)$. Therefore, $V-D'$ is a dominating set in $G$ of size 
$n-\Gamma(Q(G))$ and hence $\gamma(G)\le n-\Gamma(Q(G))$.

Now assume $D'\subseteq E$.   
Let us denote by $G'$ the subgraph of $G$ formed by the edges (and all their endpoints) $e$ such that $v_e\in D'$. Then: 

\begin{itemize}
\item {\it $G'$ is a spanning forest of $G$}, because $D'$ covers $V$ (else $D'$ is not dominating in $Q(G)$) and $G'$ is acyclic 
(else $D'$ is not a minimal dominating set in $Q(G)$). 
\item {\it $G'$ is $P_4$-free}, i.e. each connected component of $G'$ is a star, since otherwise $D'$ is not a minimal dominating set in $Q(G)$,
because any vertex of $D'$ corresponding to the middle edge of a $P_4$ in $G'$ can be removed from $D'$.
\end{itemize}

Let $D$ be the set of the centers of the stars of $G'$. Then $D$ is dominating in~$G$ (since $D'$ covers $V$) and $|D| = n-|D'|$, 
i.e. $\gamma(G)\le n-\Gamma(Q(G))$, as required.
\end{proof}

Since the {\sc minimum dominating set} problem is NP-hard and $Q(G)$ is a co-bipartite graph, Lemma~\ref{lem:Q} leads to the following conclusion.

\begin{theorem}\label{thm:2}
The {\sc upper dominating set} problem restricted to the class of complements of bipartite graphs is NP-hard. 
\end{theorem}

\section{A boundary class for {\sc upper domination}}
\label{sec:boundary}

Since the {\sc upper dominating set} problem is NP-hard in the class of complements of bipartite graphs, this class must contain a boundary class for the problem.
An idea about the structure of such a boundary class comes from Theorem~\ref{thm:domS} and Lemma~\ref{lem:Q} and can be roughly described as follows: 
a boundary class for {\sc upper domination} consists of graphs $Q(G)$ obtained from 
graphs $G$ in $\mathcal{S}$. In order to transform this idea into a formal proof, we need more notations and more auxiliary results.

For an arbitrary class $X$ of graphs, we denote $S(X):=\{S(G)\ :\ G\in X\}$ and $Q(X):=\{Q(G)\ :\ G\in X\}$. 
In particular, $Q({\cal G})$ is the set of all $Q$-graphs, where $\cal G$ is the class of all simple graphs. 
We observe that an induced subgraph of a $Q$-graph is not necessarily a $Q$-graph. Indeed, in a $Q$-graph
every new vertex is adjacent to {\it exactly} two old vertices. However, by deleting some old vertices in a $Q$-graph
we may obtain a graph in which a new vertex is adjacent to at most one old vertex. Therefore, $Q(X)$ is not necessarily hereditary even if 
$X$ is a hereditary class. We denote by $Q^*(X)$ the hereditary closure of $Q(X)$, i.e. the class obtained from $Q(X)$ 
by adding to it all induced subgraphs of the graphs in $Q(X)$. Similarly, we denote by $S^*(X)$ the hereditary closure of $S(X)$.

With the above notation, our goal is proving that $Q^*({\cal S})$ is a boundary class for 
the {\sc upper dominating set} problem. To achieve this goal we need the following lemmas.

\begin{lemma}\label{lem:monotone}
Let $X$ be a monotone class of graphs such that ${\cal S}\not\subseteq X$, then the clique-width of the graphs in $Q^*(X)$ is bounded by a constant.  
\end{lemma} 
 
\begin{proof}
In \cite{critical}, it was proved that if ${\cal S}\not\subseteq X$, then the clique-width is bounded for graphs in $X$. 
It is known (see e.g.  \cite{CO00}) that for monotone classes, the clique-width is bounded if and only if the tree-width is bounded.
By subdividing the edges of all graphs in $X$ exactly once, we transform $X$ into the class $S(X)$, where the tree-width is still bounded, since the subdivision of 
an edge of a graph does not change its tree-width. Since bounded tree-width implies bounded clique-width (see e.g.  \cite{CO00}),
we conclude that $S(X)$ is a class of graphs of bounded clique-width. Now, for each graph $G$ in $S(X)$ we create two cliques by complementing 
the edges within the sets of new and old vertices. This transforms $S(X)$ into $Q(X)$. 
It is known (see e.g. \cite{recent}) that local complementations applied finitely many times do not change the clique-width ``too much'',
i.e they transform a class of graphs of bounded clique-width into another class of graphs of bounded clique-width. Therefore, the clique-width of graphs in $Q(X)$ is bounded.
Finally, the clique-width of a graph is never smaller than the clique-width of any of its induced subgraphs (see e.g.  \cite{CO00}). Therefore,    
the clique-width of graphs in $Q^*(X)$ is also bounded.
\end{proof}

\begin{lemma}\label{lem:cw}
Let $X\subseteq Q^*({\cal G})$ be a hereditary class. The clique-width of graphs in $X$ is bounded by a constant 
if and only if it is bounded for $Q$-graphs in $X$.  
\end{lemma} 
 
\begin{proof}
The lemma is definitely true if $X=Q^*(Y)$ for some class $Y$. In this case, by definition, every non-$Q$-graph in $X$
is an induced subgraph of a $Q$-graph from $X$. However, in general, $X$ may contain a non-$Q$-graph $H$ such that no $Q$-graph containing $H$
as an induced subgraph belongs to $X$. In this case, we prove the result as follows. 

First, we transform each graph $H$ in $X$ into a bipartite graph $H'$ by replacing the two cliques of $H$ (i.e. the sets of old and new vertices) with independent sets.
In this way, $X$ transforms into a class $X'$ which is a subclass of $S^*({\cal G})$. As we mentioned in the proof of Lemma~\ref{lem:monotone},
this transformation does not change the clique-width ``too much'', i.e. the clique-width of graphs in $X$ is bounded if and only if 
it is bounded for graphs in $X'$.

By definition, $H\in X$ is a $Q$-graph if and only if $H'$ is a subdivision graph, i.e. $H'=S(G)$ for some graph $G$. 
Therefore, we need to show that the clique-width of graphs in $X'$ is bounded  
if and only if it is bounded for subdivision graphs in $X'$. In one direction, the statement is trivial.
To prove it in the other direction, assume the clique-width of subdivision graphs in $X'$ is bounded.  
If $H'$ is not a subdivision, it contains new vertices of degree 0 or 1. If $H'$ contains a vertex of degree 0, then it is disconnected,
and if $H'$ contains a vertex $x$ of degree 1, then it has a cut-point (the neighbour of $x$). 
It is well-known that the clique-width of graphs in a hereditary class is bounded if and only if it is bounded for connected graphs in the class.
Moreover, it was shown in \cite{LR}  that  the clique-width of graphs in a hereditary class is bounded if and only if it is bounded for 2-connected graphs 
(i.e. connected graphs without cut-points) in the class.
Since connected graphs without cut-points in $X'$ are subdivision graphs, we conclude that the clique-width is bounded for all graphs in $X'$.  
\end{proof}

\medskip
Finally, to prove the main result of this section, we need to show that $Q^*({\cal G})$ is a finitely defined class. 
To show this, we first characterize graphs in $Q^*({\cal G})$ as follows:
a graph $G$ belongs to $Q^*({\cal G})$ if and only if the vertices of $G$ can be partitioned into two (possibly empty) cliques $U$ and $W$ such that 
\begin{itemize} 
\item[(a)] every vertex in $W$ has at most two neighbours in $U$,
\item[(b)] if $x$ and $y$ are two vertices of $W$ each of which has {\it exactly} two neighbours in $U$, then $N(x)\cap U\ne N(y)\cap U$.
\end{itemize}
In the proof of the following lemma, we call any partition satisfying (a) and (b) {\it nice}. Therefore,    
$Q^*({\cal G})$ is precisely the class of graphs admitting a nice partition. Now we characterize $Q^*({\cal G})$ in terms of minimal 
forbidden induced subgraphs. 

\begin{figure}\centering
\begin{tikzpicture}
[scale=.5,auto=left]
\node[vertex] (a) at (0,0) {};
\node[vertex] (b) at (4,0) {};
\node[vertex] (c) at (4,4) {};
\node[vertex] (d) at (0,4) {};
\node[vertex] (e) at (1.25,2) {};
\node[vertex] (f) at (2.75,2) {};


\draw[edge] (a) -- (b);
\draw[edge] (c) -- (d);
\draw[edge] (d) -- (a);

\draw[edge] (a) -- (e);
\draw[edge] (a) -- (f);

\draw[edge] (b) -- (e);
\draw[edge] (b) -- (f);

\draw[edge] (c) -- (e);
\draw[edge] (c) -- (f);

\draw[edge] (d) -- (e);
\draw[edge] (d) -- (f);

\draw[edge] (e) -- (f);
\coordinate [label=center:$G_1$] (G1) at (2,-1);
\end{tikzpicture}
\begin{tikzpicture}
[scale=.5,auto=left]
\node[vertex] (a) at (0,0) {};
\node[vertex] (b) at (4,0) {};
\node[vertex] (c) at (4,4) {};
\node[vertex] (d) at (0,4) {};
\node[vertex] (e) at (1.25,2) {};
\node[vertex] (f) at (2.75,2) {};

\draw[edge] (b) -- (c);

\draw[edge] (a) -- (b);
\draw[edge] (c) -- (d);
\draw[edge] (d) -- (a);

\draw[edge] (a) -- (e);
\draw[edge] (a) -- (f);

\draw[edge] (b) -- (e);
\draw[edge] (b) -- (f);

\draw[edge] (c) -- (e);
\draw[edge] (c) -- (f);

\draw[edge] (d) -- (e);
\draw[edge] (d) -- (f);

\draw[edge] (e) -- (f);
\coordinate [label=center:$G_2$] (G2) at (2,-1);
\end{tikzpicture}
\begin{tikzpicture}
[scale=.5,auto=left]
\node[vertex] (a) at (0,0) {};
\node[vertex] (b) at (4,0) {};
\node[vertex] (c) at (4,4) {};
\node[vertex] (d) at (0,4) {};
\node[vertex] (e) at (1.25,2) {};
\node[vertex] (f) at (2.75,2) {};

\draw[edge] (a) -- (b);
\draw[edge] (c) -- (d);

\draw[edge] (a) -- (e);
\draw[edge] (a) -- (f);

\draw[edge] (b) -- (e);
\draw[edge] (b) -- (f);

\draw[edge] (c) -- (e);
\draw[edge] (c) -- (f);

\draw[edge] (d) -- (e);
\draw[edge] (d) -- (f);

\draw[edge] (e) -- (f);
\coordinate [label=center:$G_3$] (G3) at (2,-1);
\end{tikzpicture}
\begin{tikzpicture}
[scale=.5,auto=left]
\node[vertex] (a) at (0,0) {};
\node[vertex] (b) at (4,0) {};
\node[vertex] (c) at (4,4) {};
\node[vertex] (d) at (0,4) {};
\node[vertex] (e) at (5.5,2) {};
\node[vertex] (f) at (-1.5,2) {};

\draw[edge] (a) -- (b);
\draw[edge] (b) -- (c);
\draw[edge] (c) -- (d);
\draw[edge] (d) -- (a);
\draw[edge] (a) -- (c);
\draw[edge] (b) -- (d);

\draw[edge] (c) -- (e);
\draw[edge] (b) -- (e);

\draw[edge] (a) -- (f);
\draw[edge] (d) -- (f);
\coordinate [label=center:$G_4$] (G4) at (2,-1);
\end{tikzpicture}

\begin{tikzpicture}
[scale=.3,auto=left]
\node[vertex] (a) at (0,0) {};
\node[vertex] (b) at (4,0) {};
\node[vertex] (c) at (4,4) {};
\node[vertex] (d) at (0,4) {};
\node[vertex] (e) at (2,5.5) {};
\node[vertex] (f) at (2,-1.5) {};
\node[vertex] (g) at (-2,2) {};
\node[vertex] (h) at (6,2) {};

\draw[edge] (a) -- (b);
\draw[edge] (c) -- (d);

\draw[edge] (c) -- (e);
\draw[edge] (d) -- (e);
\draw[edge] (a) -- (f);
\draw[edge] (b) -- (f);
\draw[edge] (g) -- (c);
\draw[edge] (g) -- (d);
\draw[edge] (g) -- (a);
\draw[edge] (g) -- (b);
\draw[edge] (h) -- (c);
\draw[edge] (h) -- (d);
\draw[edge] (h) -- (a);
\draw[edge] (h) -- (b);

\draw[edge] (g) .. controls +(1,3)   ..  (e);
\draw[edge] (g) .. controls +(1,-3)  ..  (f);
\draw[edge] (h) .. controls (5,5)  .. (e);
\draw[edge] (h) .. controls (5,-1)  .. (f);
\coordinate [label=center:$G_5$] (G5) at (2,-3);
\end{tikzpicture}
\begin{tikzpicture}
[scale=.3,auto=left]
\node[vertex] (a) at (0,0) {};
\node[vertex] (b) at (4,0) {};
\node[vertex] (c) at (4,4) {};
\node[vertex] (d) at (0,4) {};
\node[vertex] (e) at (2,5.5) {};
\node[vertex] (f) at (2,-1.5) {};
\node[vertex] (g) at (-2,2) {};
\node[vertex] (h) at (6,2) {};

\draw[edge] (a) -- (b);
\draw[edge] (c) -- (d);
\draw[edge] (d) -- (a);

\draw[edge] (c) -- (e);
\draw[edge] (d) -- (e);
\draw[edge] (a) -- (f);
\draw[edge] (b) -- (f);
\draw[edge] (g) -- (c);
\draw[edge] (g) -- (d);
\draw[edge] (g) -- (a);
\draw[edge] (g) -- (b);
\draw[edge] (h) -- (c);
\draw[edge] (h) -- (d);
\draw[edge] (h) -- (a);
\draw[edge] (h) -- (b);

\draw[edge] (g) .. controls +(1,3)   ..  (e);
\draw[edge] (g) .. controls +(1,-3)  ..  (f);
\draw[edge] (h) .. controls (5,5)  .. (e);
\draw[edge] (h) .. controls (5,-1)  .. (f);
\coordinate [label=center:$G_6$] (G6) at (2,-3);
\end{tikzpicture}
\begin{tikzpicture}
[scale=.3,auto=left]
\node[vertex] (a) at (0,0) {};
\node[vertex] (b) at (4,0) {};
\node[vertex] (c) at (4,4) {};
\node[vertex] (d) at (0,4) {};
\node[vertex] (e) at (2,5.5) {};
\node[vertex] (f) at (2,-1.5) {};
\node[vertex] (g) at (-2,2) {};
\node[vertex] (h) at (6,2) {};

\draw[edge] (a) -- (b);
\draw[edge] (c) -- (d);
\draw[edge] (d) -- (a);
\draw[edge] (b) -- (c);

\draw[edge] (c) -- (e);
\draw[edge] (d) -- (e);
\draw[edge] (a) -- (f);
\draw[edge] (b) -- (f);
\draw[edge] (g) -- (c);
\draw[edge] (g) -- (d);
\draw[edge] (g) -- (a);
\draw[edge] (g) -- (b);
\draw[edge] (h) -- (c);
\draw[edge] (h) -- (d);
\draw[edge] (h) -- (a);
\draw[edge] (h) -- (b);

\draw[edge] (g) .. controls +(1,3)   ..  (e);
\draw[edge] (g) .. controls +(1,-3)  ..  (f);
\draw[edge] (h) .. controls (5,5)  .. (e);
\draw[edge] (h) .. controls (5,-1)  .. (f);
\coordinate [label=center:$G_7$] (G7) at (2,-3);
\end{tikzpicture}
\begin{tikzpicture}
[scale=.3,auto=left]
\node[vertex] (a) at (0,0) {};
\node[vertex] (b) at (4,0) {};
\node[vertex] (c) at (4,4) {};
\node[vertex] (d) at (0,4) {};
\node[vertex] (e) at (2,5.5) {};
\node[vertex] (f) at (2,-1.5) {};
\node[vertex] (g) at (-2,2) {};
\node[vertex] (h) at (6,2) {};

\draw[edge] (a) -- (b);
\draw[edge] (c) -- (d);
\draw[edge] (d) -- (a);
\draw[edge] (b) -- (c);
\draw[edge] (e) -- (f);

\draw[edge] (c) -- (e);
\draw[edge] (d) -- (e);
\draw[edge] (a) -- (f);
\draw[edge] (b) -- (f);
\draw[edge] (g) -- (c);
\draw[edge] (g) -- (d);
\draw[edge] (g) -- (a);
\draw[edge] (g) -- (b);
\draw[edge] (h) -- (c);
\draw[edge] (h) -- (d);
\draw[edge] (h) -- (a);
\draw[edge] (h) -- (b);

\draw[edge] (g) .. controls +(1,3)   ..  (e);
\draw[edge] (g) .. controls +(1,-3)  ..  (f);
\draw[edge] (h) .. controls (5,5)  .. (e);
\draw[edge] (h) .. controls (5,-1)  .. (f);
\coordinate [label=center:$G_8$] (G8) at (2,-3);
\end{tikzpicture}
	\caption{Forbidden graphs for $\Qs{{\cal G}}$}
	\label{fig:forbg}
\end{figure}
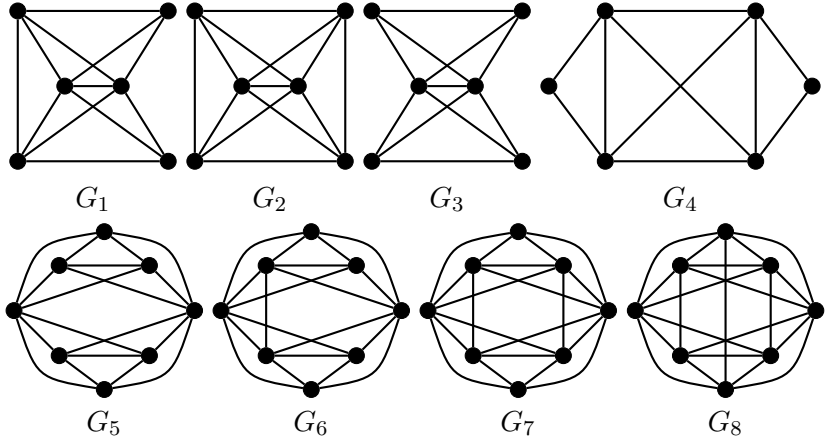

\begin{lemma}\label{lem:gamma}
$Q^*({\cal G})={\rm Free}(N)$, where $N$ is the set of eleven graphs consisting of $\overline{C}_3$, $\overline{C}_5$, $\overline{C}_7$ and the
eight graphs shown in Figure~\ref{fig:forbg}.   
\end{lemma}

\begin{proof}
To show the inclusion $Q^*({\cal G})\subseteq {\rm Free}(N)$, we first observe that $\overline{C}_3$, $\overline{C}_5$ and $\overline{C}_7$
are forbidden in $Q^*({\cal G})$, since every graph in this class is co-bipartite, while $\overline{C}_3$, $\overline{C}_5$, $\overline{C}_7$ are not
co-bipartite. Each of the remaining eight graphs of the set $N$ is co-bipartite, but none of them admits a nice partition, which is a routine matter to check. 

To prove the inverse inclusion ${\rm Free}(N) \subseteq Q^*({\cal G})$, let us consider a graph $G$ in ${\rm Free}(N)$. 
By definition, $G$ contains no $\overline{C}_3$, $\overline{C}_5$, $\overline{C}_7$. Also, since $G_1$ is an induced 
subgraph of $\overline{C}_i$ with $i\ge 9$, we conclude that $G$ contains no complements of odd cycles of length $9$ or more. 
Therefore, $G$ is co-bipartite. Let $V_1\cup V_2$
be an arbitrary bipartition of $V(G)$ into two cliques. In order to show that $G$ belongs to $Q^*({\cal G})$, we split our analysis into several cases.

\medskip
{\it Case} 1: $G$ contains a $K_4$ induced by vertices $x_1,y_1\in V_1$ and $x_2,y_2\in V_2$. To analyze this case, we partition the vertices of $V_1$ into 
four subsets with respect to $x_2$, $y_2$ as follows:
\begin{itemize}
\item[$A_1$] is the set of vertices of $V_1$ adjacent to $x_2$ and non-adjacent to $y_2$,
\item[$B_1$] is the set of vertices of $V_1$ adjacent to $x_2$ and to $y_2$,
\item[$C_1$] is the set of vertices of $V_1$ adjacent to $y_2$ and non-adjacent to $x_2$,
\item[$D_1$] is the set of vertices of $V_1$ adjacent neither to $x_2$ nor to $y_2$.
\end{itemize}
We partition the vertices of $V_2$ with respect to $x_1$, $y_1$ into four subsets $A_2$, $B_2$, $C_2$, $D_2$ analogously.  We now observe the following.

\begin{itemize}
\item[(1)] For $i\in \{1,2\}$, either $A_i=\emptyset$ or $C_i=\emptyset$, since otherwise a vertex in $A_i$ and a vertex in $C_i$ together 
with $x_1,y_1,x_2,y_2$ induce $G_2$.
\end{itemize}
According to this observation, in what follows, we may assume, without loss of generality, that 
\begin{itemize}
\item $C_1=\emptyset$ and $C_2=\emptyset$.
\end{itemize}
We next observe that
\begin{itemize}
\item[(2)] Either $A_1=\emptyset$ or $A_2=\emptyset$, since otherwise a vertex $a_1\in A_1$ and a vertex $a_2\in A_2$ together 
with $x_1,y_1,x_2,y_2$ induce either $G_1$ (if $a_1$ is not adjacent to $a_2$) or $G_2$ (if $a_1$ is adjacent to $a_2$).
\end{itemize}
Observation (2) allows us to assume, without loss of generality, that 
\begin{itemize}
\item $A_2=\emptyset$.
\end{itemize} 
We further make the following conclusions:
\begin{itemize}
\item[(3)] For $i\in \{1,2\}$, $|D_i|\le 1$, since otherwise any two vertices of $D_i$ together with $x_1,x_2,y_1,y_2$ induce $G_3$.
\item[(4)] If $D_1=\{d_1\}$ and $D_2=\{d_2\}$, then $d_1$ is adjacent to $d_2$, since otherwise $d_1,d_2,x_1,x_2$, $y_1$, $y_2$ induce $G_4$.
\item[(5)] If $A_1\cup D_1\cup D_2\ne\emptyset$, then every vertex of $B_1$ is adjacent to every vertex of $B_2$. Indeed, assume, without loss of generality, 
that $z\in A_1\cup D_1$ and a vertex $b_1\in B_1$ is not adjacent to a vertex $b_2\in B_2$. Then the vertices $z,b_1,b_2,x_1,x_2,y_1$ induce either $G_1$ 
(if $z$ is not adjacent to $b_2$) or $G_2$ (if $z$ is adjacent to $b_2$).
\item[(6)] Either $A_1=\emptyset$ or $D_1=\emptyset$, since otherwise a vertex in $A_1$ and a vertex in $D_1$ together 
with $x_1,y_1,x_2,y_2$ induce $G_1$.
\end{itemize}
According to (6), we split our analysis into three subcases as follows.

\medskip
{\it Case} 1.1: $D_1=\{d_1\}$. Then $A_1=\emptyset$ (by (6)) and every vertex of $B_1$ is adjacent to every vertex of $B_2$ (by (5)).
If $D_2=\emptyset$, then $U=D_1$ and $W=B_1\cup B_2$ is a nice partition of $G$ (remember that $x_1,y_1\in B_1$ and $x_2,y_2\in B_2$). 

Now assume $D_2=\{d_2\}$ and denote by $B_1^0$ the vertices of $B_1$
nonadjacent to $d_2$ and by $B_1^1$ the vertices of $B_1$ adjacent to $d_2$. Similarly, we denote by $B_2^0$ the vertices of $B_2$
nonadjacent to $d_1$ and by $B_2^1$ the vertices of $B_2$ adjacent to $d_1$. Then $|B_1^1\cup B_2^1|\le 1$, since otherwise 
any two vertices of $B_1^1\cup B_2^1$ together with $x_1,x_2,d_1,d_2$ induce $G_2$. But then $U=D_1\cup D_2$ and $W=B_1\cup B_2$ is a nice partition of $G$.

\medskip
{\it Case} 1.2: $A_1\ne\emptyset$. Then $D_1=\emptyset$ (by (6)) and every vertex of $B_1$ is adjacent to every vertex of $B_2$ (by (5)).
In this case, we claim that 
\begin{itemize}
\item[(7)] every vertex of $B_2$ is either adjacent to every vertex of $A_1$ or to none of them. Indeed, assume a vertex $b_2\in B_2$ 
has a neighbour $a'\in A_1$ and a non-neighbour $a''\in A_1$. Then $b_2,a',a'',x_1,y_1,y_2$ induce $G_1$.
\end{itemize} 
We denote by $B_2^0$ the subset of vertices of $B_2$ that have no neighbours in $A_1$ and by $B_2^1$ the subset of vertices of $B_2$ 
adjacent to every vertex of $A_1$. Then 
\begin{itemize}
\item either $|A_1|= 1$ or $|B_2^0|= 1$, since otherwise any two vertices of $A_1$ together with any two vertices 
of $B_2^0$ and any two vertices of $B_1$ induce $G_3$. 
\item if $D_2=\{d_2\}$, then $|B_2^1|= 1$, since otherwise any two vertices of $B_2^1$ together with $d_2,x_1,y_2$ and any vertex $a$ in $A_1$
induce either $G_1$ (if $a$ is not adjacent to $d_2)$) or $G_2$ (if $a$ is adjacent to $d_2)$).
\item if $D_2=\{d_2\}$, then $d_2$ has no neighbours in $B_1$. Indeed, if  $d_2$ has a neighbour $b_1\in B_1$, then vertices $b_1,d_2,x_1,x_2,y_2$
together with any vertex $a_1\in A_1$ induce either $G_1$ (if $d_2$ is not adjacent to $a_1$) or $G_2$ (if $d_2$ is adjacent to $a_1$).  
\end{itemize}
Therefore, either $U=A_1\cup D_2$, $W=B_1\cup B_2$ (if $|A_1|=1$) or $U=B_2^0\cup D_2$, $W=A_1\cup B_1\cup B_2^1$ (if $|B_2^0|=1$) is a nice partition of $G$.

\medskip
{\it Case} 1.3: $A_1=\emptyset$ and $D_1=\emptyset$. In this case, if $D_2\ne\emptyset$, then $U=D_2$, $W=B_1\cup B_2$ is a nice partition of $G$, 
since $B_1\cup B_2$ is a clique (by (5)).
Assume now that $D_2=\emptyset$. If $B_1\cup B_2$ is a clique, then $G$ has a trivial nice partition. Suppose next that $B_1\cup B_2$ is not a clique.
If all non-edges of $G$ are incident to a same vertex, say $b$ (i.e. $b$ is incident to all the edges of $\overline{G}$),
then $U=\{b\}$, $W=(B_1\cup B_2)-\{b\}$ is a nice partition of $G$. Otherwise, $G$ contains a pair of non-edges $b_1'b_2'\not\in E(G)$ and $b_1''b_2''\not\in E(G)$
with all four vertices $b_1',b_1''\in B_1$, $b_2',b_2''\in B_2$ being distinct (i.e. $b_1'b_2'$ and $b_1''b_2''$ form a matching in $\overline{G}$). 
We observe that $\{b_1',b_1'',b_2',b_2''\}\cap \{x_1,y_1,x_2,y_2\}=\emptyset$, because by definition vertices $x_1,y_1,x_2,y_2$ dominate the set $B_1\cup B_2$.
But then  $b_1',b_1'',b_2',b_2'',x_1,y_1$ induce either $G_2$ (if both $b_1'b_2''$ and $b_2'b_1''$ are edges in $G$) or 
$G_1$ (if exactly one of $b_1'b_2''$ and $b_2'b_1''$ is an edge in $G$) or $G_3$ (if neither $b_1'b_2''$ nor $b_2'b_1''$ is an edge in $G$).
This completes the proof of Case~1.

\medskip
{\it Case} 2: $G$ contains no $K_4$ with two vertices in $V_1$ and two vertices in $V_2$. We claim that in this case $V_1\cup V_2$ is a nice partition of $G$.
First, the assumption of case 2 implies that that no two vertices in the same part of the bipartition $V_1\cup V_2$ have two common neighbours in the opposite part,
verifying condition (b) of the definition of nice partition. To verify condition (a), it remains to prove that one of the parts $V_1$ and $V_2$ has no vertices 
with more than two neighbours in the opposite part. Assume the contrary and let $a_1\in V_1$ have three neighbours in $V_2$ and let $a_2\in V_2$ have three neighbours in $V_1$.

First, suppose $a_1$ is adjacent to $a_2$. Denote by $b_2,c_2$ two other neighbours of $a_1$ in $V_2$ and by $b_1,c_1$ two other neighbours of $a_2$ in $V_1$.
Then there are no edges between $b_1,c_1$ and $b_2,c_2$, since otherwise we are in conditions of Case 1. But now $a_1,b_1,c_1,a_2,b_2,c_2$ induce a $G_3$.

Suppose now that  $a_1$ is not adjacent to $a_2$. We denote by $b_2,c_2,d_2$ three neighbours of $a_1$ in $V_2$ and by $b_1,c_1,d_1$ three neighbours of $a_2$ in $V_1$.
No two edges between $b_1,c_1,d_1$ and $b_2,c_2,d_2$ (if any) share a vertex,  since otherwise we are in conditions of Case 1. But then $a_1,b_1,c_1,d_1,a_2,b_2,c_2,d_2$ 
induce either $G_5$ or $G_6$ or $G_7$ or $G_8$. This contradiction completes the proof of the lemma. 
\end{proof}

\medskip
Now we are ready to prove the main result of the section.

\begin{theorem}\label{thrm:boundary}
If $\mathrm{P}\neq \mathrm{NP}$, then $Q^*(\mathcal{S})$ is a boundary class for the {\sc upper dominating set} problem.
\end{theorem}

\begin{proof}
From Lemmas~\ref{lem:Zk} and~\ref{lem:Q} we know that {\sc upper domination} is NP-hard in the class $Q^*(Z_k)$
for all values of $k\ge 3$. Also, it is not hard to verify that the sequence of classes $Q^*(Z_1),Q^*(Z_2)\ldots$ converges to $Q^*({\cal S})$.
Therefore, $Q^*({\cal S})$ is a limit class for the {\sc upper dominating set} problem. To prove its minimality, 
assume there is a limit class $X$ which is properly contained in $Q^*(\cal S)$. We consider a graph 
$F\in Q^*({\cal S})-X$, a graph $G\in Q({\cal S})$ containing $F$ as an induced subgraph (possibly $G=F$ if $F\in Q({\cal S})$) and
a graph $H\in {\cal S}$ such that $G=Q(H)$.  
From the choice of $G$ and Lemma~\ref{lem:gamma}, we know that $X \subseteq \Free(N\cup \{G\})$, where $N$ is the set of 
minimal forbidden induced subgraphs for the class $Q^*({\cal G})$. Since the set $N$ is finite (by Lemma~\ref{lem:gamma}), 
we conclude with the help of Lemma~\ref{lem:limit} that the {\sc upper dominating set} problem is NP-hard in the class $\Free(N\cup \{G\})$.
To obtain a contradiction, we will show that graphs in $\Free(N\cup \{G\})$ have bounded clique-width.

Denote by $M$ the set of all graphs containing $H$ as a spanning subgraph. Clearly $\Free(M)$ is a monotone class. 
More precisely, it is the class of graphs containing no $H$ as a subgraph (not necessarily induced). 
Since $\Free(M)$ is monotone and $\cal S\not\subset \Free(M)$ (as $H\in {\cal S}$), we know from Lemma~\ref{lem:monotone}
that the clique-width is bounded in $\Qs{\Free(M)}$. 

To prove that graphs in $\Free(N\cup \{G\})$ have bounded clique-width, we will show that $Q$-graphs in this class belong to  $\Qs{\Free(M)}$.
Let $Q(H')$ be a $Q$-graph in $\Free(N\cup \{G\})$. Since the vertices of $Q(H')$ represent the vertices and the edges of $H'$
and $Q(H')$ does not contain $G$ as an induced subgraph, we conclude that $H'$ does not contain $H$ as a subgraph. 
Therefore, $H'\in \Free(M)$, and hence $Q(H')\in Q(\Free(M))$. By Lemma~\ref{lem:cw}, this implies that all graphs  
in $\Free(N\cup \{G\})$ have bounded clique-width. This contradicts the fact that the {\sc upper dominating set} problem is NP-hard in this class
and completes the proof of the theorem.
\end{proof}

\section{A dichotomy in monogenic classes}
\label{sec:poly}

The main goal of this section is to show that in the family of monogenic classes the {\sc upper dominating set} problem
admits a dichotomy, i.e. for each graph $H$, the problem is either polynomial-time solvable or NP-hard for $H$-free graphs. 
We start with polynomial-time results.

\subsection{Polynomial-time results}

As we have mentioned in the introduction, the {\sc upper dominating set} problem can be 
solved in polynomial time for bipartite graphs \cite{CockayneFPT81}, chordal graphs \cite{JacobsonP90} 
and generalized series-parallel graphs \cite{HareHLPPW87}. It also admits a polynomial-time 
solution in any class of graphs of bounded clique-width \cite{Courcelle}. 
Since $P_4$-free graphs have clique-width at most 2 (see e.g. \cite{cliqie-width}), we make the following conclusion.

\begin{proposition}\label{obs:1}
The {\sc upper dominating set} problem can be solved for $P_4$-free graphs in polynomial time.  
\end{proposition}

In what follows, we develop a polynomial-time algorithm to solve the problem in the class of $2K_2$-free graphs. 

We start by observing that the class of $2K_2$-free graphs admits a polynomial-time 
solution to the {\sc maximum independent set} problem (see e.g. \cite{2K2}). By Lemma~\ref{lem:1}
every maximal (and hence maximum) independent set is a minimal dominating set. These observations 
allow us to restrict ourselves to the analysis of minimal dominating sets $X$ such that 
\begin{itemize}
\item $X$ contains at least one edge,
\item $|X|>\alpha(G)$,
\end{itemize}
where $\alpha(G)$ is the independence number, i.e. the size of a maximum independent set in $G$.

Let $G$ be a $2K_2$-free graph and let $ab$ be an edge in $G$. Assuming that $G$ contains a minimal dominating 
set $X$ containing both $a$ and $b$, we first explore some properties of $X$. In our analysis we use
the following notation. We denote by
\begin{itemize}
\item  $N$ the neighbourhood of $\{a,b\}$, i.e. the set of vertices outside of 
$\{a,b\}$ each of which is adjacent to at least one vertex of $\{a,b\}$,
\item $A$ the anti-neighbourhood of $\{a,b\}$, i.e. the set of vertices adjacent neither to $a$ nor to $b$,
\item $Y:=X\cap N$,  
\item $Z:=N(Y)\cap A$, i.e. the set of vertices of $A$ each of which is adjacent to at least one vertex of $Y$.
\end{itemize}

Since $a$ and $b$ are adjacent, by Lemma~\ref{lem:1} each of them has a private neighbour outside of $X$.
We denote by 
\begin{itemize}
\item $a^*$ a private neighbour of $a$, 
\item $b^*$ a private neighbour of $b$. 
\end{itemize}
By definition, $a^*$ and $b^*$ belong to $N-Y$ and have no neighbours in $Y$. 
Since $G$ is $2K_2$-free, we conclude that

\medskip
\noindent 
{\it Claim} 1.
$A$ is an independent set.

\medskip
We also derive a number of other helpful claims. 

\medskip
\noindent
{\it Claim} 2.
$Z\cap X = \emptyset$ and $A-Z \subseteq X$.

\begin{proof}
Assume a vertex $z\in Z$ belongs  to $X$. Then $X-\{z\}$ is a dominating set, because  $z$ does not dominate any vertex of 
$A$ (since $A$ is independent) and it is dominated by its neighbor in $Y$. This contradicts the minimality of $X$
and proves that $Z\cap X = \emptyset$. Also, by definition, no vertex of  $A-Z$ has a neighbour in $Y\cup \{a,b\}$. Therefore, 
to be dominated $A-Z$ must be included in $X$. 
\end{proof}

\medskip
\noindent 
{\it Claim} 3.
If $|X| > \alpha(G)$, then $|Y| = |Z|$ and every vertex of $Z$ is a private neighbor of a vertex in $Y$.

\begin{proof}
Since every vertex $y$ in $Y$ belongs to $X$ and has a neighbour in $X$ ($a$ or $b$), by Lemma~\ref{lem:1} $y$ must have a private neighbor in $Z$.
Therefore, $|Z| \geq |Y|$. If $|Z|$ is strictly greater than $|Y|$, then $|X|\le |A\cup \{a\}|\le \alpha(G)$ (since $A$ is independent),
which contradicts the assumption  $|X| > \alpha(G)$. Therefore, $|Y| = |Z|$ and every vertex of $Z$ is a private neighbor of a vertex in $Y$.
\end{proof}

\medskip
\noindent 
{\it Claim} 4.
If $|Y| > 1$ and $|X| > \alpha(G)$, then $Y \subseteq N(a) \cap N(b)$.

\begin{proof}
Let $y_1, y_2$ be two vertices in $Y$ and let $z_1, z_2$ be two vertices in $Z$ which are private neighbours of $y_1$ and $y_2$, respectively. 

Assume $a$ is not adjacent to $y_1$, then $b$ is adjacent to $y_1$ (by definition of $Y$) and $a^*$ is adjacent to $z_1$, 
since otherwise the vertices $a, a^*, y_1, z_1$ induce a $2K_2$ in $G$. Also, $a^*$ is adjacent to $z_2$, since otherwise 
a $2K_2$ is induced by $a^*,z_1,y_2,z_2$. But now the vertices $a^*, z_2,b,y_1$ induce a $2K_2$. This contradiction shows that $a$ is adjacent to $y_1$. 
Since $y_1$ has been chosen arbitrarily, $a$ is adjacent to every vertex of $Y$, and by symmetry, $b$ is adjacent to every vertex of $Y$.
\end{proof}

\medskip
\noindent 
{\it Claim} 5.
If $|Y| > 1$ and $|X| > \alpha(G)$, then $a^*$ and $b^*$ have no neighbours in~$Z$.

\begin{proof}
Assume by contradiction that $a^*$ is adjacent to a vertex $z_1\in Z$. By Claim~3, 
$z_1$ is a private neighbour of a vertex $y_1\in Y$. Since $|Y| > 1$, there exists another vertex $y_2\in Y$ 
with a private neighbor $z_2\in Z$. From Claim~4, 
we know that $b$ is adjacent to $y_2$.
But then the set $\{b,y_2,a^*,z_1\}$ induces a $2K_2$. This contradiction shows that $a^*$ has no neighbours in~$Z$.
By symmetry, $b^*$ has no neighbours in in~$Z$.
\end{proof}

The above series of claims leads to the following conclusion, which plays a key role for the development of 
a polynomial-time algorithm. 

\begin{lemma}\label{lemma:main1}
If $|X| > \alpha(G)$, then $|Y|=1$ and $Y \subseteq N(a) \cap N(b)$. 
\end{lemma}

\begin{proof}
First, we show that $|Y|\le 1$. Assume to the contrary that $|Y| > 1$.
By definition of $a^*$ and Claim~2, 
vertex $a^*$ has no neighbours in $A-Z$, and by Claim~5, 
$a^*$ has no neighbours in $Z$. Therefore, $A\cup \{a^*,b\}$ is 
an independent set of size $|X|=|Y|+|A-Z|+2$. This contradicts the assumption that 
$|X| > \alpha(G)$ and proves that $|Y|\le 1$.

Suppose now that $|Y|=0$. Then, by Claim~3, $|Z|=0$ and hence, by Claim~2, $X=A\cup \{a,b\}$.
Also, by definition of $a^*$, vertex $a^*$ has no neighbours in $A$. But then $A\cup \{a^*,b\}$ is 
an independent set of size $|X|$, contradicting that $|X| > \alpha(G)$.

From the above discussion we know that $Y$ consists of a single vertex, say $y$. It remains to show that
$y$ is adjacent to both $a$ and $b$. By definition, $y$ must be adjacent to at least one of them, say to $a$. 
Assume that $y$ is not adjacent to $b$. By definition of $a^*$, vertex $a^*$ has no neighbours in $\{y\}\cup (A-Z)$,
and by definition of $Z$, vertex $y$ has no neighbours in $A-Z$. But then $(A-Z)\cup \{a^*,b,y\}$ is an independent 
set of size $|X|=|Y|+|A-Z|+2$. This contradicts the assumption that 
$|X| > \alpha(G)$ and shows that $y$ is adjacent to both $a$ and $b$.
\end{proof}

\begin{corollary}\label{cor}
If a minimal dominating set in a $2K_2$-free graph $G$ is larger than $\alpha(G)$, 
then it consists of a triangle and all the vertices not dominated by the triangle.   
\end{corollary}

In what follows, we describe an algorithm~$\mathcal{A}$ to find a minimal dominating set $M$ with maximum cardinality in a $2K_2$-free 
graph $G$ in polynomial time. In the description of the algorithm, given a graph $G=(V,E)$ and a subset $U\subseteq V$, we denote by $A(U)$
the anti-neighbourhood of $U$, i.e. the subset of vertices of $G$ outside of $U$ none of which has a neighbour in $U$. 

\medskip
\noindent
{\it Algorithm}~$\mathcal{A}$

\begin{flushleft}
{\bf Input:} A $2K_2$-free graph $G=(V,E)$.\\
{\bf Output:} A minimal dominating set $M$ in $G$ with maximum cardinality.
\end{flushleft}
\begin{enumerate}
	\item Find a maximum independent set $M$ in $G$.  

	\item  For each triangle $T$ in $G$:
	\begin{itemize}
		\item Let $M' := T \cup A(T)$.
		\item If $M'$ is a minimal dominating set and $|M'| > |M|$, then $M := M'$.
	\end{itemize}
	\item Return $M$.
\end{enumerate}

\begin{theorem}\label{thm:poly}
Algorithm~$\mathcal{A}$ correctly solves the {\sc upper dominating set} problem for $2K_2$-free graphs in polynomial time.  
\end{theorem}

\begin{proof}
Let $G$ be a $2K_2$-free graph with $n$ vertices. In $O(n^2)$ time, one can find a maximum independent set $M$ in $G$ (see e.g. \cite{2K2}).
Since $M$ is also a minimal dominating set (see Lemma \ref{lem:mm}), any solution of size at most $\alpha(G)$ can be ignored.

If $X$ is a solution of size more than  $\alpha(G)$, then, by Corollary~\ref{cor}, 
it consists of a triangle $T$ and its anti-neighbourhood $A(T)$. For each triangle $T$, verifying whether $T \cup A(T)$
is a minimal dominating set can be done in $O(n^2)$ time. Therefore, the overall time complexity of the algorithm can be 
estimated as $O(n^5)$.
\end{proof}

\subsection{The dichotomy}
\label{sec:main}

In this section, we summarize the results presented earlier in order to obtain the following dichotomy.

\begin{theorem}
Let $H$ be a graph. If $H$ is a $2K_2$ or $P_4$ (or any induced subgraph of $2K_2$ or $P_4$), then 
the {\sc upper dominating set} problem can be solved for $H$-free graphs in polynomial time. 
Otherwise the problem is NP-hard for $H$-free graphs.
\end{theorem}

\begin{proof}
Assume  $H$ contains a cycle $C_k$, then the problem is NP-hard for $H$-free graphs  
\begin{itemize}
\item either by Theorem~\ref{thm:1} if $k\le 5$, because in this case the class of $H$-free graphs contains all 
graphs of girth at least 6,  
\item or by Theorem~\ref{thm:2} if $k\ge 6$, because in this case the class of $H$-free graphs contains 
the class of $\overline{K}_3$-free graphs and hence all complements of bipartite graphs.  
\end{itemize}
Assume now that $H$ is acyclic, i.e. a forest. If it contains a claw (a star whose center has degree 3), 
then the problem is NP-hard for $H$-free graphs by Theorem~\ref{thm:2}, because in this case the class of $H$-free graphs 
contains all $\overline{K}_3$-free graphs and hence all complements of bipartite graphs.

If $H$ is a claw-free forest, then every connected component of $H$ is a path. If $H$ contains at least three connected components, 
then the class of $H$-free graphs contains all $\overline{K}_3$-free graphs, in which case the problem is NP-hard by Theorem~\ref{thm:2}.
Assume $H$ consists of two connected components $P_k$ and $P_t$. 
\begin{itemize}
\item If $k+t\ge 5$, then the class of $H$-free graphs contains 
all $\overline{K}_3$-free graphs and hence the problem is NP-hard by Theorem~\ref{thm:2}. 
\item If $k+t\le 3$, then the class of $H$-free graphs is a subclass of $P_4$-free graphs and 
hence the problem can be solved in polynomial time in this class by Proposition~\ref{obs:1}. 
\item If $k+t=4$, then 
\begin{itemize}
\item either $k=t=2$,  in which case $H=2K_2$ and hence the problem can be solved in polynomial 
time by Theorem~\ref{thm:poly}, 
\item or $k=4$ and $t=0$, in which case $H=P_4$ and hence the problem can be solved 
in polynomial time by Proposition~\ref{obs:1}, 
\item or $k=3$ and $t=1$, in which case the class of $H$-free graphs 
contains all $\overline{K}_3$-free graphs and hence the problem is NP-hard by Theorem~\ref{thm:2}.
\end{itemize}
\end{itemize} 

\end{proof}

\section{Conclusion}
\label{sec:con}

In this paper, we identified the first boundary class for the {\sc upper dominating set} problem
and proved that the problem admits a dichotomy for monogenic classes, i.e. classes defined by 
a single forbidden induced subgraph. We conjecture that this dichotomy can be extended to all 
finitely defined classes. By Theorem~\ref{thm:boundary}, the problem is NP-hard in a finitely defined 
class $X$ if and only if $X$ contains a boundary class for the problem. In the present paper, we made
the first step towards the description of the family of boundary classes for {\sc upper domination}.
Since the problem is NP-hard in the class of triangle-free graphs (Theorem~\ref{thm:1}), there must exist 
at least one more boundary class for the problem. We believe that this is again the class $\cal S$ of tripods. This class 
was proved to be boundary for many algorithmic graph problems, which is typically done by showing that 
a problem is NP-hard in the class ${\cal Z}_k$ for any fixed value of $k$. We believe that the same is true for
{\sc upper domination}, but this question remains open. 

One more open question deals with Lemma~\ref{lem:Q} of the present paper. It shows a relationship 
between {\sc minimum dominating set} in general graphs and {\sc upper dominating set} in co-bipartite graphs. 
For the first of these problems, three boundary classes are available \cite{AKL04}. One of them
was transformed in the present paper to a boundary class for {\sc upper domination}. Whether the other two can 
also be transformed in a similar way is an interesting open question, which we leave for future research.

\section*{Acknowledgements}
Vadim Lozin and Viktor Zamaraev gratefully acknowledge support from EPSRC, grant EP/L020408/1.

Part of this research was carried out when Vadim Lozin was visiting the King Abdullah University of Science and Technology (KAUST).
This author thanks the University for hospitality and stimulating research environment.

\end{document}